\documentclass{sig-alternate}

\usepackage[utf8]{inputenc}
\usepackage{cmap}
\usepackage{microtype}
\usepackage{lmodern}
\usepackage{amsthm}
\usepackage{amssymb,amsmath}
\usepackage{graphics}
\usepackage{amsfonts}
\usepackage{mathrsfs}
\usepackage{amscd}
\usepackage{comment}
\usepackage{color}
\usepackage{url}
\usepackage{bbm}
\usepackage{tikz}
\usepackage[plainpages=false,pdfpagelabels,colorlinks=true,citecolor=blue,hypertexnames=false]{hyperref}

\newtheorem{theorem}{Theorem}
\newtheorem{prop}[theorem]{Proposition}

\newtheorem{lemma}[theorem]{Lemma}
\newtheorem{algorithm}[theorem]{Algorithm}

\def\qed{\quad\rule{1ex}{1ex}}

\let\set\mathbbm

\def\e{\mathrm{e}}
\def\tuple#1{\boldsymbol{#1}}

\def\Exp{\operatorname{Exp}}
\def\Re{\operatorname{Re}}
\def\Im{\operatorname{Im}}

\begin{document}

\title{Finding Hyperexponential Solutions of Linear ODEs by Numerical Evaluation}

\numberofauthors{3}

\author{%
 \alignauthor Fredrik Johansson\titlenote{Supported by the Austrian Science Fund (FWF) grant Y464-N18.}\\[\medskipamount]
      \affaddr{\strut RISC}\\
      \affaddr{\strut Johannes Kepler University}\\
      \affaddr{\strut 4040 Linz, Austria}\\[\smallskipamount]
      \email{\strut fjohanss@risc.jku.at}
 \alignauthor Manuel Kauers\raisebox{.43em}{\normalsize$^{\textstyle\ast}$}\\[\medskipamount]
      \affaddr{\strut RISC}\\
      \affaddr{\strut Johannes Kepler University}\\
      \affaddr{\strut 4040 Linz, Austria}\\[\smallskipamount]
      \email{\strut mkauers@risc.jku.at}
 \alignauthor Marc Mezzarobba\\[\medskipamount]
 \affaddr{\strut Inria, Univ. Lyon, AriC, LIP\smash{\titlenote{UMR 5668 CNRS -- ENS Lyon -- Inria -- UCBL}}}\\
      \affaddr{\strut ENS de Lyon, 46 allée d'Italie}\\
      \affaddr{\strut 69364 Lyon Cedex 07, France}\\[\smallskipamount]
      \email{\strut marc@mezzarobba.net}
}

\maketitle

\begin{abstract}
  We present a new algorithm for computing hyperexponential solutions of ordinary linear differential equations
  with polynomial coefficients. 
  The algorithm relies on interpreting formal series solutions at the singular points as analytic functions
  and evaluating them numerically at some common ordinary point. 
  The numerical data is used to determine a small number of combinations of the formal series that may
  give rise to hyperexponential solutions. 
\end{abstract}




\overfullrule=5pt

\section{Introduction}

We consider linear differential operators 
\[
  P = p_r D^r + p_{r-1} D^{r-1} + \cdots + p_0
\]
where $p_0,\dots,p_r$ are polynomials and $D$ represents the standard derivation
$\frac{d}{dx}$.  Such operators act in a natural way on elements of a
differential ring containing the polynomials. An object $y$ is called a solution
of the operator if $P$ applied to $y$ yields zero. We are interested in finding
the hyperexponential solutions of a given operator. An object $y$ is called
hyperexponential if the quotient $D(y)/y$ can be identified with a rational
function. Typical examples are rational functions (e.g.\ $(5x+3)/(3x+5)$),
radicals (e.g.\ $\sqrt{x+1}$), exponentials (e.g.\ $\exp(3x^2-4)$ or $\exp(1/x)$),
or combinations of these (e.g.\ $\sqrt{x+1}\exp(x^9/(x-1))$).  Equivalently,
$y$~is called hyperexponential if there is some first order operator $q_1 D +
q_0$ with $q_0,q_1$ polynomials which maps $y$ to zero. If we regard 
differential operators as elements of an operator algebra $C(x)[D]$, then there is a one-to-one
correspondence between the hyperexponential solutions $y$ of an operator~$P$ and
its first order right hand factors. In other words, if $y$ is a hyperexponential term with
$(q_1 D + q_0)\cdot y=0$, then $y$ is a solution of $P$ if and
only if there exist rational functions $u_0,\dots,u_{r-1}$ such that
\[
 P = (u_{r-1}D^{r-1}+u_{r-2}D^{r-2}+\cdots+u_0)(q_1D+q_0).
\]
Algorithms for finding the hyperexponential solutions of a linear differential
equation (or equivalently, the first order right hand factors of the
corresponding operators) are known since long. They are needed as subroutine in
algorithms for factoring operators or for finding Liouvillean solutions. See
Chapter~4 of~\cite{put03} for details and references.

Classical algorithms first compute ``local solutions'' at singular points (cf.\
Section~\ref{sec:2:3} below) and then test for each combination of local solutions whether it gives rise to
a hyperexponential solution. This leads to a combinatorial explosion with
exponential runtime. The situation is similar to classical algorithms for
factoring polynomials over~$\set Q$, which first compute the irreducible factors
modulo a prime and then test for each combination whether it gives rise to a
factor in~$\set Q[x]$.

The algorithm of van Hoeij~\cite{hoeij97a} avoids the combinatorial explosion as
follows. It picks one local solution and considers the operator $Q=q_1D+q_0$
with $q_1,q_0\in C((x))$ which annihilates it. This operator is a right factor
of~$P$, though not with rational coefficients. The algorithm then constructs (if
possible) a left multiple $B$ of $Q$ with rational coefficients of order at most
$r-1$. This leads to a nontrivial factorization $P=AB$ in $C(x)[D]$. The
procedure is then applied recursively to $A$ and $B$ until a complete factorization
is found. The first order factors in this factorization give rise to at most $r$
hyperexponential candidate solutions (possibly up to multiplication by a
rational function). These are then checked in a second step. Van Hoeij's algorithm reminds of the
polynomial factorization algorithm of Lenstra, Lenstra, Lovász~\cite{lenstra82,vzgathen99},
which picks one modular factor and constructs (if possible) a multiple of this
factor with integer coefficients but smaller degree than the original
polynomial. This multiple is then a proper divisor in~$\set Q[x]$. 

The algorithm we propose below avoids the combinatorial explosion in a different
way. We start from the local solutions and regard them as asymptotic expansions
of complex functions. By means of effective analytic continuation and arbitrary-precision 
numerical evaluation, we compute the values of these functions at
some common ordinary reference point. Then a linear algebra algorithm is used to
determine a small list of possible combinations of local solutions that
may give rise to hyperexponential ones, possibly up to multiplication by a rational
function. These are then checked in a second step. Our approach was motivated by
van Hoeij's polynomial factorization algorithm~\cite{hoeij02}, which associates to every
modular factor a certain vector and then uses lattice reduction to determine a small 
list of combinations that may give rise to proper factors. 

Although our algorithm avoids the combinatorial explosion problem, we do not
claim that it runs in polynomial time. Indeed, no polynomial time algorithm can
be expected because there are operators $P$ which have hyperexponential
solutions~$y$ that are exponentially larger than~$P$. 
Also van Hoeij~\cite{hoeij97a}
makes no formal statement about the complexity of his algorithm. It is clear
though that his algorithm is superior to the naive algorithm. Similarly, we
believe that our algorithm has chances to outperform van Hoeij's algorithm, at
least in examples that are not deliberately designed to exhibit worst case
performance. The reason is partly that during the critical combination phase we
only work with floating point numbers of moderate precision while van Hoeij's 
algorithm in general needs to do arithmetic in algebraic number fields whose 
degrees may grow during the computation. 
Another advantage of our algorithm is that it is conceptually 
simpler than van Hoeij's, at least if we take for granted that we can compute
high-precision evaluations of D-finite functions.

\section{Preliminaries}

In this section, we recall some results from the literature and introduce
notation that will be used in subsequent sections.

\subsection{Differential Fields and Operator Algebras}

A differential ring/field is a pair $(K,D)$ where $K$ is a ring/field and
$D\colon K\to K$ is a derivation on~$K$, i.e., a map satisfying $D(a+b)=D(a)+D(b)$
and $D(ab)=D(a)b+aD(b)$ for all $a,b\in K$. Throughout this paper, we consider
the differential field $K=C(x)$, where $C$ is some (computable) subfield of~$\set C$,
together with the derivation $D\colon K\to K$ defined by $D(c)=0$ for all $c\in C$
and $D(x)=1$. For simplicity, we assume throughout that $C$ is algebraically closed. 

A differential ring/field $E$ is called an extension of $K$ if $K\subseteq E$, and the
derivation of~$E$ restricted to $K$ agrees with the derivation of~$K$. 

By $K[D]$ we denote the set of all polynomials in the indeterminate $D$ with
coefficients in~$K$. Addition in~$K[D]$ is defined in the usual way, and
multiplication is defined subject to the commutation rule $D a = a D + D(a)$ for
$a\in K$. The elements of $K[D]$ are called operators, and they act on the
elements of some extension $E$ of $K$ in the obvious way: If $P=p_0+p_1 D
+ \cdots + p_r D^r$ is an operator of order~$r$ and $y\in E$, 
then $P\cdot y:=\sum_{i=1}^r p_i D^i(y)\in E$. The noncommutative multiplication
is compatible with operator application in the sense that we have 
$(PQ)\cdot y=P\cdot(Q\cdot y)$ for all $P,Q\in K[D]$ and all~$y\in E$.

The elements $y\in E$ such that $P\cdot y=0$ form a $C$-vector space~$V$ with
$\dim V\leq r$. By making $E$ sufficiently large it can always be assumed
that $\dim V= r$.

\subsection{Hyperexponential Terms}

Let $E$ be an extension of $K$. An element $h\in E\setminus\{0\}$ is called
\emph{hyperexponential} over $K$ if $D(h)/h\in K$.  Equivalently, $h$ is
hyperexponential if $Q\cdot h=0$ for some nonzero first order operator~$Q\in
K[D]$.

Two hyperexponential terms $h_1,h_2$ are called \emph{equivalent} if $h_1/h_2\in K$. 
For example, the terms $\exp(3x^2-x)$ and $(1-2x)^2\exp(3x^2-x)$ are equivalent, but $\exp(3x^2-x)$ and 
$(1-2x)^{\sqrt 2}\exp(3x^2-x)$ are not. (Here and below, we use standard calculus notation 
to refer to elements of some extension~$E$ on which the derivation acts as the 
notation suggests, e.g.\ $D(\exp(3x^2-x)) = (6x-1)\exp(3x^2-x)$.)

Every hyperexponential term can be written in the form $h=\exp(\int v)$, where
$v$ is a rational function. The additive constant of the integral amounts to a multiplicative
constant for~$h$, which is irrelevant in our context, because $P\cdot h=0$ if and only if $P\cdot(c h)=0$
for every $c\in C\setminus\{0\}$. If we consider the partial fraction decomposition of $v$ and
integrate it termwise, we obtain something of the form
\[
  g + \sum_{i=1}^n \gamma_i \log(p_i)
\]
with $g\in K$, $\gamma_1,\dots,\gamma_n\in C$ and monic square free
pairwise coprime polynomials $p_i\in C[x]$. In
terms of this representation, two hyperexponential terms are equivalent if the
difference of the corresponding rational functions~$g$ is a constant and 
any two corresponding coefficients $\gamma_i$ differ by an integer.

The equivalence class of a hyperexponential term $h$ is called the
\emph{exponential part} of~$h$. The motivation for this terminology is that when
we are searching for some hyperexponential solution $h$ of $P$ and we already
know its equivalence class, then we can take an arbitrary element $h_0$ from
this class and make an ansatz $h=u h_0$ for some rational function $u\in
K$. The operator 
$\tilde P:=P\otimes \bigl(D-\frac{D(1/h_0)}{1/h_0}\bigr)\in K[D]$
then has
the property that $u$ is a solution of $\tilde P$ if and only if $uh_0$
is a solution of~$P$. This reduces the problem to finding rational solutions,
which is well understood and will not be discussed here~\cite{abramov91,put03}.

\subsection{Local Solutions}\label{sec:2:3}\label{sec:localsol}

Consider an operator $P\in C(x)[D]$ of order~$r$. By clearing denominators, if necessary, we
may assume that $P\in C[x][D]$, say $P=p_r D^r + \cdots + p_0$ with $p_r\neq0$.
A point $z\in\set C\cup\{\infty\}$ is called \emph{singular} if $z$ is a root of~$p_r$,
or $z=\infty$. A point which is not singular is called \emph{ordinary.} Note that there are
only finitely many singular points, and that we include the ``point at infinity'' always
among the singular points. 

If $z=0$ is an ordinary point then $P$ admits $r$ linearly independent power series
solutions. If $z=0$ is a singular point, it is still possible to find $r$ linearly
independent generalized series solutions of the form 
\begin{alignat}1 \label{eq:1}
  x^{\alpha}\exp(u(x^{-1/s})) \sum_{k=0}^m b_k(x^{1/s}) \log(x)^k
\end{alignat}
where $\alpha\in C$, $u\in C[x]$ with $u(0)=0$, $s\in\set N$, $m\in\set N$ and 
$b_0,\dots,b_m\in C[[x]]$.
We call these solutions the \emph{local solutions} at~$0$. The computation of such solutions
is well-known and will not be discussed here~\cite{hoeij97,put03}.

Two series as in \eqref{eq:1} are called equivalent if they have the same $u$ and~$s$
and the difference of the respective values of $\alpha$ is in $\frac1s\set Z$. The
equivalence classes of generalized series under this equivalence relation are
called the \emph{exponential parts} of the series. Adopting van Hoeij's notation
and defining $\Exp(e):=\exp(\int\frac ex)$ for $e\in C[x^{-1/s}]$, we have
that $\Exp(e_1)$ and $\Exp(e_2)$ are equivalent iff $e_1-e_2\in\frac1s\set Z$. 
Note that if $m=0$ and $s=1$, two series are equivalent iff their quotient can be identified
with a formal Laurent series.
We will from now on make no notational distinction between $\Exp(e)$ and its equivalence class. 

A point $z\neq0$ can be moved to the origin by the change of variables
$\tilde x=x-z$ (if $z\in C$) or $\tilde x=1/x$ (if $z=\infty$). If $\tilde P$ 
is the operator obtained from $P$ by replacing $x$ by $\tilde x+z$ or
$1/\tilde x$, then a local solution of $P\in C[x][D]$ at $z$ is defined as the
local solution of $\tilde P\in C[\tilde x][D]$ at~$0$.

Throughout the rest of this paper, we will use the following notation.
$P$~is some operator in $C[x][D]$ of order~$r$, by $z_1,\dots,z_{n-1}\in C$
we denote its finite singular points, $z_n=\infty$. We write $\tilde x_i=x-z_i$ ($i=1,\dots,n-1$)
and $\tilde x_n=1/x$ for the variables with respect to which the singularities at $z_i$
appear at the origin. For $i=1,\dots,n$, we consider
the vector space~$V_i$ generated by all local solutions at~$z_i$. There may be solutions with different exponential
parts, say $\ell_i$ different parts $\Exp(e_{i,1}),\dots,\Exp(e_{i,\ell_i})$ for
$e_{i,j}\in C[\tilde x_i^{-1/s_{i,j}}]$. By 
\[
  V_{i,j} = V_i \cap \Exp(e_{i,j}) C((\tilde x_i^{1/s_{i,j}}))[\log\tilde x_i]
\] 
we denote the vector space of all local solutions of $P$ at $z_i$ with exponential part
(equivalent to) $\Exp(e_{i,\ell_i})$. 
Our $V_{i,j}$ are written $V_{e_{i,j}}(P)$ in van Hoeij's papers~\cite{hoeij97,hoeij97a}.

The condition in the definition of equivalence that the difference of corresponding values
of $\alpha$ be an integer (rather than, say, requiring exactly the same value of~$\alpha$)
ensures that the $V_{i,j}$ are indeed vector spaces, because if some $V_{i,j}$ contains, 
for example, the two series
\begin{alignat*}1
  x^\alpha (1 + x + x^2 + \cdots)\quad\text{and}\quad x^\alpha (1 + x + 3x^2 + \cdots)
\end{alignat*}
then it must also contain their difference $x^\alpha( 2x^2 + \cdots) = x^{\alpha+2}(2+\cdots)$.

\subsection{Analytic Solutions}\label{sec:analytic}

It is classical that the formal power series solutions $\hat y$ of $P$ at an
ordinary point $z\in\set C$ actually converge in a neighbourhood of $z$ and thus
give rise to analytic function solutions~$y$ of~$P$. The correspondence is
one-to-one. For any other ordinary point $z'\in\set C$ and a path $z\leadsto z'$ 
avoiding singular points there exists a matrix $M_{z\leadsto z'}\in\set C^{r\times r}$
such that 
\[
  \bigl(D^j y(z')\bigr)_{j=0}^{r-1} = M_{z\leadsto z'}\bigl(D^j y(z)\bigr)_{j=0}^{r-1}
\]
for every solution $y$ analytic near~$z$.
There are algorithms~\cite{ChudnovskyChudnovsky1990,vdHoeven1999} for
efficiently computing the entries of $M_{z\leadsto z'}$ for any given polygon
path $z\leadsto z'$ with vertices in $\bar{\set Q}$ to any desired
precision. In other words, we can compute arbitrary precision approximations of
$y$ and its derivatives at every ordinary point (``effective analytic continuation'').

Assume now that $0$ is a singular point, and consider the case $s=1$ and $m=0$, i.e., 
let $\hat y=\Exp(e)b$ for some $e\in C[x^{-1}]$ and $b\in C[[x]]$ be a formal solution of~$P$.
To give an analytic meaning to $\Exp(e)=\exp(\int\frac ex)=\exp(u+\alpha\log x)=x^\alpha\exp(u)$ 
(for suitable $\alpha\in C$ and $u\in C[x^{-1}]$) amounts to making a choice for a branch of
the logarithm. Every choice gives rise to the same function up to some multiplicative constant.

Since $\Exp(e)b$ is a solution of $P$ iff $b$ is a solution of the operator
$P\otimes (D+\frac ex)$, we may assume that $e=0$. Then the problem remains that the
formal power series $\hat y=b$ may not be convergent if $0$ is a singular
point. However, by resummation theory \cite{Balser1994,Balser2000} it
is still possible to associate to $\hat y$ an analytic function~$y$
defined on some sector
\[
 \Delta=\Delta(d,\varphi,\rho):=\{z\in\set C: 0<|z|\leq\rho\land |d-\arg z|\leq\varphi/2\}
\]
(with $d\in[0,2\pi]$, $\rho,\varphi>0$) such that $\hat y$ is the asymptotic expansion of 
$y$ for $z\to0$ in~$\Delta$.

The precise formulation of this result is technical and not really needed for
our purpose (see \cite[Chap.~6, 10, and 11]{Balser2000} or
\cite[Chap.~5--7]{Balser1994} for full details). It will be more than
sufficient to know the following facts:
\begin{itemize}
\item For every
  $\tuple k=(k_1,\dots,k_q)\in\set Q^q$
  with \hbox{$k_1 > \dots > k_q$} and every
  $\tuple d=(d_1,\dots,d_q)\in[0,2\pi]^q$
  such that
  \[ | d_{j+1} - d_j| \leq (k_{j+1}^{-1} - k_j^{-1}) \tfrac{\pi}{2}, 
  \quad j=1, \dots, q-1, \]
  one constructs \cite[§10.2]{Balser2000} a differential subring
  $\set C\{x\}_{\tuple k,\tuple d}$ of $\set C[[x]]$
  \cite[Theorems 51 and 53]{Balser2000}
  which contains the ring
  $\set C\{x\}$ of all convergent power series.
\item There is a differential ring homomorphism
  \cite[Theorems 51 and 53]{Balser2000}
  $\mathcal{S}_{\tuple k,\tuple d}$
  from $\set C\{x\}_{\tuple k,\tuple d}$ to the germs of analytic
  functions defined on sectors of the form $\Delta(d_1,\varphi,\rho)$ for
  suitable $\varphi,\rho>0$,
  with the property that for every $\hat y\in\set C\{x\}_{\tuple
    k,\tuple d}$ the function $\mathcal{S}_{\tuple k,\tuple d}(\hat y)$ has
  $\hat y$ as its asymptotic expansion for $z\to0$
  \cite[§10.2, Exercice 2]{Balser2000}.
  The $\mathcal{S}_{\tuple k,\tuple d}$ map convergent formal
  power series to their sum in the usual sense \cite[Lemmas 8 and 20]{Balser2000}.
%
\item \label{item:Skd-domain}
  For a given operator $P\in C[x][D]$ of order~$r$, one can compute a
  tuple $\tuple k$ and finite subsets $\mathcal D_1, \dots, \mathcal D_q$
  of~$[0, 2\pi]$ such that any $\hat y\in\set C[[x]]$ with $P\cdot\hat y=0$
  belongs to $\set C\{x\}_{\tuple k,\tuple d}$ for all $\tuple d$ as above with
  $d_1 \notin \mathcal D_1, \dots, d_q \notin \mathcal D_q$.
  Additionally, given such a $\tuple d$, one can compute $\varphi, \rho >0$ such that
  each $\mathcal{S}_{\tuple k,\tuple d}(\hat y)$ is defined on
  $\Delta(d_1,\varphi,\rho)$.
\item Furthermore, given a point $z\in
  \Delta(d_1,\varphi,\rho)$, a precision $\varepsilon>0$, and
  $\hat y\in\set C[[x]]$ with $P\cdot\hat y=0$, one can efficiently compute an
  approximation~$Y_\varepsilon$ of the vector $Y(z) = (D^j \mathcal{S}_{\tuple k,\tuple d}(\hat
  y))_{j=0}^{r-1}$ such that $\|Y(z) - Y_\varepsilon\| \leq \varepsilon$.
\end{itemize}

The computational part of the last two items is a special case of Theorem~7 of van der Hoeven~\cite{hoeven07}. 
As an application, van der Hoeven~\cite{vdhoeven07} shows how to factor differential operators using numerical evaluation.
Note that our $k_j$ correspond to $1/k_j$ in van der Hoeven's articles, and the components of the tuples $\tuple k$ and $\tuple d$ appear in reverse order.

Also observe that in the last item, $z$ is an ordinary point, so that from there we can use effective analytic
continuation to compute values of $\mathcal{S}_{\tuple k,\tuple d}(\hat y)$ and its derivatives at any
other ordinary point. 

\section{Outline of the Algorithm}

A hyperexponential term~$h$ can be expanded as a generalized series at every point $z\in\set C\cup\{\infty\}$,
in particular at its singularities. The resulting generalized series are local solutions of~$P$
if $h$ is a solution of~$P$. If $h=\exp(\int v)$ is a hyperexponential 
solution where $v\in\set C(x)$, and if we write the partial fraction decomposition of $v$ in the form 
\[
  v = \frac{e_1}{x-z_1} + \frac{e_2}{x-z_2} + \cdots + \frac{e_n}{1/x},
\]
where the $e_i$ are polynomials in $\tilde x_i^{-1}$, then expanding this $h$ at
$z_i$ yields a generalized series in $\tilde x_i$ whose exponential part matches~$\Exp(e_i)$. The components
$e_i$ in the decomposition of $v$ must hence show up among the exponential parts of
the local solutions of~$P$.

If $\Exp(e_{i,1}),\dots,\Exp(e_{i,\ell_i})$ are (representatives of) the different exponential parts that appear among the 
local solutions at~$z_i$, then any hyperexponential solution must be equivalent to the
term $\exp (\int (\frac{e_{1,j_1}}{\tilde x_1}+\cdots+\frac{e_{n,j_n}}{\tilde x_n}))$ for some tuple $(j_1,\dots,j_n)$. 
It then remains to check for each of these candidates whether some element of its equivalence class
solves the given equation. 
The basic structure of the algorithm for finding hyperexponential solutions is thus as follows.

\begin{algorithm} \label{alg:main}
  \textit{Input:} a linear differential operator $P=p_0+p_1D+\cdots+p_rD^r$, $p_r \neq 0$, with coefficients in $C[x]$.\\
  \textit{Output:} all the hyperexponential terms $h$ with $P\cdot h=0$.

  \kern-\smallskipamount
  \begin{enumerate}
  \item Let $z_1,\dots,z_{n-1}\in\set C$ be the roots of $p_r$ in $\set C$, and let $z_n=\infty$. 

  \kern-\smallskipamount
  \item For $i=1,\dots,n$ do

  \kern-\smallskipamount
  \item \quad\vtop{\hsize=.865\hsize Find the exponential parts $\Exp(e_{i,1}), \dots, \Exp(e_{i,\ell_i})$
      of the local solutions of $P$ at~$z_i$.}

  \kern-\smallskipamount
  \item\label{alg:2:4} Determine a set $U\subseteq\{1,\dots,\ell_1\}\times\cdots\times\{1,\dots,\ell_n\}$ 
    s.t. for every hyperexponential solution~$h$ equivalent to 
    $\exp\bigl(\int\sum_{i=1}^n\frac{e_{i,j_i}}{\tilde x_i}\bigr)$ we have $(j_1,\dots,j_n)\in U$.

  \kern-\smallskipamount
  \item For each $(j_1,\dots,j_n)\in U$ do
  
  \kern-\smallskipamount
  \item\quad\vtop{\hsize=.865\hsize Let $h_0:=\exp\bigl(\int\sum_{i=1}^n\frac{e_{i,j_i}}{\tilde x_i}\bigr)$, 
      and compute the operator $\tilde P:=P\otimes (D-\frac{D(1/h_0)}{1/h_0})$.}
  
  \kern-\smallskipamount
  \item\quad\vtop{\hsize=.865\hsize Compute a basis $\{u_1,\dots,u_m\}\subseteq C(x)$
      of the vector space of all rational solutions of~$\tilde P$, and output
      $u_1h_0$, \dots, $u_mh_0$.}
  \end{enumerate}

\end{algorithm}  

There is some freedom in step~\ref{alg:2:4} of this algorithm. 
A naive approach would simply be to take all possible combinations, i.e.,
$U=\{1,\dots,\ell_1\}\times\cdots\times\{1,\dots,\ell_n\}$. 
This is a finite set, but its size is in general exponential in the number of singular points. 
For finding a smaller set~$U$, Cluzeau and van Hoeij~\cite{cluzeau04} use modular techniques
to quickly discard unnecessary tuples. 
Our algorithm, explained in the following section, addresses the same issue. 
It computes a set $U$ of at most~$r$ tuples.

\section{The Combination Phase}
\label{sec:combination}

In general, the differential operator $P$ may have several different solutions
with the same exponential part, i.e., the dimension of the vector spaces
$V_{i,j}$ might be greater than one. In this case, it might be that $V_{i,j}$
contains some series which is the expansion of a hyperexponential solution~$h$
at~$z_i$ as well as some other series which are not. If we compute some basis of~$V_{i,j}$,
we cannot expect it to contain the expansion of~$h$. Instead, each basis element will in
general be the linear combination of this series and some other one. 
Now, if the expansion of $h$ at some other singular point $z_{i'}$ belongs to
the space $V_{i',j'}$ (which possibly also has higher dimension), then, in
some sense, $h$~must belong to the intersection of the vector spaces $V_{i,j}$
and~$V_{i',j'}$.

Our algorithm is based on testing which intersections are nontrivial.
To make these intersections meaningful, we must first map the vector
spaces we want to intersect into a common ambient space~$W$. Let $E$ be some
differential ring containing $C(x)$ as well as all the hyperexponential solutions of~$P$, 
and let $W\subseteq E$ be the $C$-vector space generated by solutions of $P$ in~$E$. 
For each $i$, let $\pi_i$ be some vector space homomorphism
\[
  \bigoplus_{j=1}^{\ell_i} \Exp(e_{i,j})\set C((\tilde x_i^{1/s_{i,j}}))[\log\tilde x_i]
  \supseteq V_i
  \stackrel{\pi_i}\longrightarrow W
\]
with the following properties:
\begin{enumerate}
  \item \label{item:combination-1} The sum $\pi_i(V_{i,1})+\cdots+\pi_i(V_{i,\ell_i})$ is direct.
  \item \label{item:combination-2} If $h\in W$ is hyperexponential, then $\pi_i^{-1}(h)$ contains the formal series expansion $\hat h$ of $h$ at~$z_i$, possibly up to a multiplicative constant.
\end{enumerate}
Define $W_{i,j}:=\pi_i(V_{i,j})$. If $h$ is some hyperexponential 
solution of~$P$, say with exponential part 
\[
\exp\Bigl(\int \Bigl(\frac{e_{1,j_1}}{\tilde x_1}+\frac{e_{2,j_2}}{\tilde x_2}+\cdots+\frac{e_{n,j_n}}{\tilde x_n}\Bigr)\Bigr),
\]
then $h\in W_{i,j_i}$ for all~$i$, and hence the vector space $W_{1,j_1}\cap
\cdots\cap W_{n,j_n}$ is not the zero subspace (because it contains at
least~$h$). Our main observation is that there can be at most $r$ tuples $\tuple
j = (j_1,\dots,j_n)$ for which $W_{\tuple j} \neq \{0\}$, and that they can be computed
efficiently once we have bases of the~$W_{i,j}$.

Postponing the discussion of making the $\pi_i$ constructive to the next section,
assume for the moment that $W$ is some vector space over~$C$, let $r=\dim W<\infty$ be its dimension, 
and suppose we are given $n$ different decompositions of subspaces of $W$ into direct sums:
\begin{alignat*}1
   &W_{1,1} \oplus W_{1,2} \oplus\cdots\oplus W_{1,\ell_1}\subseteq W, \\
   &W_{2,1} \oplus W_{2,2} \oplus\cdots\oplus W_{2,\ell_2}\subseteq W, \\
   &\qquad\vdots\\
   &W_{n,1} \oplus W_{n,2} \oplus\cdots\oplus W_{n,\ell_n}\subseteq W.      
\end{alignat*}
Without loss of generality, we may make the following assumptions:
\begin{itemize}
\item Each direct sum $\bigoplus_{i=1}^{\ell_i} W_{i,j}$ is in fact equal to~$W$. If not, add one more
  vector space to the sum.
\item $\ell_1=\ell_2=\cdots=\ell_n=:\ell$. If not, pad the sum with several copies of $\{0\}$.
\item $\ell\leq r$. If not, then because the sums are supposed to be direct, each decomposition must
  contain at least $\ell-r$ copies of $\{0\}$, which can be dropped.
\end{itemize}
  
\begin{lemma}\label{lem:1}
  There are at most $\dim W=r$ different tuples 
  \[
    \tuple j=(j_1,\dots,j_n)\in\{1,\dots,\ell\}^n
  \]
  such that $W_{\tuple j} := W_{1,j_1}\cap W_{2,j_2}\cap\cdots\cap W_{n,j_n}\neq\{0\}$.
\end{lemma}

\begin{proof}  
  Induction on $n$. For $n=1$, there are only $\ell\leq r$ different
  tuples altogether: $(1), (2), \dots, (\ell)$, so the claim is obviously
  true. Suppose now that the claim is shown for the case when $n-1$
  decompositions of some vector space are given.
  Let $U \subset \{1,\dots,\ell\}^n$ be a set of tuples~$\tuple j$ with
  $W_{\tuple j}\neq\{0\}$. Partition the elements of~$U$ according
  to their first components,
  \[
    U = U_1 \stackrel.\cup U_2  \stackrel.\cup \cdots  \stackrel.\cup U_\ell,
  \]
  i.e., $U_k$ is the set of all tuples $\tuple j$ whose first component is~$k$, for $k=1,\dots,\ell$.

  For all $\tuple j=(k,j_2,\dots,j_n)\in U_k$ we have $\{0\}\neq W_{\tuple j}\subseteq W_{1,k}$. 
  Therefore, $(j_2,\dots,j_n)\in \{1,\dots,\ell\}^{n-1}$ is a valid solution tuple
  for the modified problem with $W'_{i,j}:=W_{i+1,j}\cap W_{1,k}$ ($i=1,\dots,n-1$, $j=1,\dots,\ell$)
  in place of $W_{i,j}$ ($i=1,\dots,n$, $j=1,\dots,\ell$). 
  By induction hypothesis, since the $W'_{i,j}$ form $n-1$ decompositions of the space~$W_{1,k}$,
  there are at most $\dim W_{1,k}$ tuples $(j_2,\dots,j_n)$ with $W_{(j_2,\dots,j_m)}\neq\{0\}$.
  Consequently, there are altogether at most $\sum_{k=1}^\ell \dim W_{1,k}=\dim W = r$
  different tuples for the original space~$W$. 
\end{proof}

The desired index tuples can be computed efficiently using dynamic programming, as shown in the following algorithm.

\begin{algorithm}\label{alg:1}\label{alg:intersect}
  \textit{Input:} a vector space $W$ of dimension~$r$, and a collection of subspaces $W_{i,j}$  
  ($i=1,\dots,n$; $j=1,\dots,\ell$) such that $W=\bigoplus_{j=1}^\ell W_{i,j}$ for $i=1,\dots,n$ and $\ell\leq r$.\\
  \textit{Output:} the set $U$ of all tuples $\tuple j=(j_1,\dots,j_n)$ with the property 
  $W_{\tuple j}=\bigcap_{i=1}^n W_{i,j_i}\neq\{0\}$.

  \kern-\smallskipamount
  \begin{enumerate}
  \item\label{alg:1:1} $U := \{\,(j): W_{1,j}\neq\{0\}\,\}$

  \kern-\smallskipamount
  \item For $i=2,\dots,n$ do

  \kern-\smallskipamount
  \item \quad $U_{\mathit{new}} := \emptyset$ 

  \kern-\smallskipamount
  \item \quad For $j=1,\dots,\ell$ do

  \kern-\smallskipamount
  \item \qquad For $\tuple k\in U$ do

  \kern-\smallskipamount
\item\label{alg:1:6}\label{step:intersect} \qquad\quad If $W_{\tuple k}\cap W_{i,j}\neq\{0\}$ then

  \kern-\smallskipamount
  \item\label{alg:1:7} \qquad\qquad $U_{\mathit{new}} := U_{\mathit{new}}\cup\{\operatorname{append}(\tuple k,j)\}$

  \kern-\smallskipamount
  \item\label{alg:1:8} \quad $U:=U_{\mathit{new}}$

  \kern-\smallskipamount
  \item Return $U$
  \end{enumerate}
\end{algorithm}

\begin{theorem} \label{thm:algo tuples}
  Algorithm~\ref{alg:1} is correct and needs no more than $8nr^4$ operations in~$C$,
  if the bases of the $W_{\tuple k}$ are cached. 
\end{theorem}

\begin{proof}
  Correctness is obvious by line~\ref{alg:1:6} and the fact that whenever $\tuple k=(k_1,\dots,k_n)$ is
  such that $W_{\tuple k}\neq\{0\}$ then we necessarily also have $W_{(k_1,\dots,k_{n-1})}\neq\{0\}$.

  For the complexity, we first show that it is a loop invariant that
  $W_{\tuple k_1}\cap W_{\tuple k_2}=\{0\}$ for any two distinct $\tuple k_1,\tuple k_2\in U$. 
  This is clear for $i=1$ by line~\ref{alg:1:1} and the assumption in the algorithm 
  specification that $W=\bigoplus_{j=1}^\ell W_{1,j}$ is a direct sum. 
  Assume it is true for some~$i$ and consider the situation right before line~\ref{alg:1:8}. 
  At this point, for any two distinct tuples $\tuple k_1,\tuple k_2\in U$ we have 
  $W_{\tuple k_1}\cap W_{\tuple k_2}=\{0\}$
  by induction hypothesis. We have to show that the same is true for any two distinct tuples 
  $\tuple k_1,\tuple k_2\in U_{\mathit{new}}$. By line~\ref{alg:1:7}, any such tuples have the form 
  $\tuple k_1=(\tuple u_1, j_1)$, $\tuple k_2=(\tuple u_2, j_2)$ for some $\tuple u_1,\tuple u_2\in U$ 
  and $j_1,j_2\in\{1,\dots,\ell\}$.
  The tuples $\tuple k_1,\tuple k_2$ are distinct if $\tuple u_1\neq\tuple u_2$ or $j_1\neq j_2$. 
  If $\tuple u_1\neq\tuple u_2$, then by induction hypothesis $W_{\tuple u_1}\cap W_{\tuple u_2}=\{0\}$, 
  and therefore also
  \begin{alignat*}1
     W_{\tuple k_1}\cap W_{\tuple k_2} 
      &= (W_{\tuple u_1}\cap W_{i,j_1})\cap (W_{\tuple u_2}\cap W_{i,j_2})\\
      &= \{0\}\cap W_{i,j_1}\cap W_{i,j_2}=\{0\}.
  \end{alignat*}
  Similarly, if $j_1\neq j_2$, then $W_{i,j_1}\cap W_{i,j_2}=\{0\}$ by the assumption 
  that $W=\bigoplus_{j=1}^\ell W_{i,j}$ is a direct sum. Therefore
  \begin{alignat*}1
     W_{\tuple k_1}\cap W_{\tuple k_2} 
      &= (W_{\tuple u_1}\cap W_{i,j_1})\cap (W_{\tuple u_2}\cap W_{i,j_2})\\
      &= W_{\tuple u_1}\cap W_{\tuple u_2}\cap \{0\}=\{0\}.
  \end{alignat*}
  This completes the proof of the loop invariant $\tuple k_1\neq\tuple k_2\Rightarrow W_{\tuple k_1}\cap W_{\tuple k_2}=\{0\}$. 

  A consequence of this invariant is that $\sum_{\tuple k\in U}\dim W_{\tuple k}\leq r$ in every iteration. 
  Since the sum $W=\bigoplus_{j=1}^\ell W_{i,j}$ is direct, we also have $\sum_{j=1}^\ell \dim W_{i,j}\leq r$
  in every iteration. 
  The intersection of two subspaces of $W$ of dimensions $d_1,d_2$ can be computed using 
  no more than 
  \[
    \min(r,d_1+d_2)^2\max(r,d_1+d_2)
  \]
  operations in~$C$.
  For the total cost of the algorithm we therefore obtain, writing $U_i$ for the set $U$ 
  in the $i$th iteration and $d_{\tuple k}:=\dim W_{\tuple k}$ and $d_{i,j}:=\dim W_{i,j}$,
  \begin{alignat*}1
          &\sum_{i=2}^n\sum_{j=1}^\ell \sum_{\tuple k\in U_i} 
            \underbrace{\min(r,d_{\tuple k}+d_{i,j})^2}_{\leq (d_{\tuple k} + d_{i,j})^2} 
            \underbrace{\max(r,d_{\tuple k}+d_{i,j})}_{\leq 2r}\\
    \leq{}&2r \sum_{i=2}^n\sum_{j=1}^\ell \sum_{\tuple k\in U_i} 
            \bigl(d_{\tuple k}^2 + 2d_{\tuple k}d_{i,j} + d_{i,j}^2\bigr)\\
    \leq{}&2r \sum_{i=2}^n\sum_{j=1}^\ell \bigl(r^2 + 2 r^2 d_{i,j} + r d_{i,j}^2\bigr)\\
    \leq{}&2r \sum_{i=2}^n \bigl(\ell r^2 + 2 r^3 + r^3\bigr)\\
    \leq{}&8nr^4.
  \end{alignat*}
  In the second step, we have used the bounds $\sum_{\tuple k\in U_i} d_{\tuple k}^2\leq r^2$ and
  $|U_i|\leq r$, which follow from $\sum_{\tuple k\in U_i} d_{\tuple k}\leq r$ and Lemma~\ref{lem:1}, respectively. 
  In the third step, we used the bound $\sum_{j=1}^\ell d_{i,j}^2\leq r^2$, which follows from
  $\sum_{j=1}^\ell d_{i,j}\leq r$. 
\end{proof}

If the objective is just to show that the algorithm runs in polynomial time, a
simpler argument applies.  It suffices to observe that all the intersections can
be done with a number of operations which is at most cubic in~$r$, then taking
also into account that we always have $|U|\leq r$ by Lemma~\ref{lem:1}, the
bound $\mathrm{O}(n \ell r^4)=\mathrm{O}(n r^5)$ follows immediately.

\section{Numerical Evaluation at a\hskip0ptplus1fill\break  Reference Point}\label{sec:5}

We now turn to the question of how to construct the morphisms~$\pi_i$. The basic idea is
to choose a reference point~$z_0$ that is an ordinary point of~$P$, and let~$W$ be the 
space of analytic solutions of the equation in a neighborhood of~$z_0$.

\hangindent=-.3\hsize\hangafter=-8
\leavevmode\smash{\rlap{\kern.71\hsize\raisebox{-6.5\baselineskip}{%
\begin{tikzpicture}
  \fill[gray] (.5, .5) circle (1.5mm);
  \fill[black] (.5, .5) circle (1pt) node[anchor=south] {\strut\raisebox{.1ex}{$z_0$}};
  \fill[gray] (-.1,0) -- +(-5:4mm) arc (-5:30:4mm) -- cycle;
  \fill[black] (-.1,0) circle (1pt) node[anchor=east] {$z_1$};
  \draw (-.1,0) -- (.45, .1) -- (.5, .5); 
  \fill[gray] (1,.3) -- +(35:3mm) arc (35:85:3mm) -- cycle;
  \fill[black] (1,.3) circle (1pt) node[anchor=north] {$z_2$};
  \draw (1,.3) -- (1.1, .5) -- (1, .75) -- (.5, .5); 
  \fill[gray] (0,.9) -- +(105:5mm) arc (105:125:5mm) -- cycle;
  \fill[black] (0,.9) circle (1pt) node[anchor=east] {$z_3$};
  \draw (0,.9) -- (-.15,1.2) -- (.1, 1.1) -- (.3, .5) -- (.5, .5); 
  \fill[gray] (.5,-.6) -- +(205:4mm) arc (205:255:4mm) -- cycle;
  \fill[black] (.5,-.6) circle (1pt) node[anchor=west] {$z_4$};
  \draw (.5,-.6) -- (.3,-.8) -- (.1, -.3) -- (.8,0) -- (.5, .5); 
\end{tikzpicture}%
}}}%
For each singular point~$z_i$, let $\Delta_i$ be a sector rooted at $z_i$ for
which all formal power series appearing in the generalized series solutions of
$P$ at $z_i$ admit an interpretation as analytic functions via some operator
$\mathcal S_{\tuple k, \tuple d}$ (depending on~$i$, but not on the
series), as described in Section~\ref{sec:analytic}.
Such sectors exist and can be computed explicitly.
Next, let $\gamma_i$ ($i=1,\dots,n$) be
polygonal paths from $z_i$ to $z_0$ avoiding singular points and leaving the
startpoint through $\Delta_i$ (meaning that for some $\varepsilon>0$ all the points on
$\gamma_i$ with a distance to $z_i$ less than $\varepsilon$ should belong
to~$\Delta_i$). Such paths exist. The analytic interpretations of the generalized series solutions
at the singular points~$z_i$ defined in $\Delta_i$ admit a unique analytic continuation
along the paths $\gamma_i$ to the neighborhood of~$z_0$.

We define $\pi_i\colon V_i \to W$ as follows. Let $V_{i,j}^0$ be the subspace of $V_{i,j}$ consisting of generalized series~\eqref{eq:1} with $s=1$ and $m=0$, and let $V'_{i,j}$ be a linear complement of~$V_{i,j}^0$ in~$V_{i,j}$.
If $\hat y\in V_{i,j}^0$ i.e., if $\hat y=\Exp(e_{i,j})b$ with $e_{i,j}\in\set
C[\tilde x_i^{-1}]$ and $b\in\set C[[\tilde x_i]]$, define $\pi_i(\hat y)$ to
be the unique analytic continuation of the function
$\operatorname{E}(e_{i,j})\mathcal{S}_{\tuple k,\tuple d}(b)$ along $\gamma_i$ to~$z_0$,
where $\operatorname{E}(e_{i,j})$ refers to the function $z\mapsto\exp(\int_{z_0}^z
e_{i,j}/\tilde x_i)$ with some arbitrary but fixed choice of the branch of the
logarithm, and $\mathcal{S}_{\tuple k,\tuple d}$ is as described in
Section~\ref{sec:analytic}.
Set $\pi_i(\hat y)=0$ for $\hat y \in V'_{i,j}$, and then extend $\pi_i$
to~$V_i$ by linearity.
The precise values of $\pi_i(V_{i,j})$ depend on the choice of $\Delta_i$ and
$\tuple d$ (which is arbitrary, within the limits indicated in
Section~\ref{sec:analytic}), but, as shown below, the properties of these spaces
used in the algorithm do not.

\begin{prop}
  The functions $\pi_i$ defined above satisfy the two requirements imposed in Section~\ref{sec:combination}:
  (1)~$\pi_i(V_{i,1})+\cdots+\pi_i(V_{i,\ell_i})$ is a direct sum; 
  (2)~if $h$~is a hyperexponential term, then $\pi_i^{-1}(h)$ contains the
  formal series expansion of $h$ at $z_i$,
  possibly up to a multiplicative constant.
\end{prop}
\begin{proof}
  1.
  Without loss of generality, we assume $z_i=0$.
  Let $\hat y_j\in V_{i,j}$ ($j=1,\dots,\ell_i$) and
  consider $\hat y=\sum_{j=1}^{\ell_i}\hat y_j$.
  Write
  $\hat y_j=x^{\alpha_j}\exp(u_j)b_j + \hat y_j'$
  where $\hat y_j' \in V_{i,j}'$,
  the $(\alpha_j,u_j)$ are pairwise distinct,
  $u_j(0)=0$,
  and $b_j(0)\neq0$ unless the series $b_j$ is zero.
  Writing $u_j=\sum_k u_{j,k}x^{-k}$, choose a direction $\theta$ such that
  $\rho \mathrm e^{\mathrm{i}\theta} \in \Delta_i$ for small $\rho$ and 
  any two unequal $u_{j, k}^{1/k} \mathrm e^{\mathrm{i}\theta}$ have different
  real parts.
  
  By changing $x$ to $\mathrm e^{-\mathrm{i}\theta} x$,
  we can assume that $d = 0$.
  This tranforms $u_j$ into
  $\sum_k (u_{j,k} \mathrm e^{\mathrm i k \theta}) x^{-k}$,
  so that the real parts of two polynomials $u_j$ can be the same only if
  the~$u_j$ themselves are equal.
  Hence, we can reorder the nonzero terms in the expression of~$\hat y$ by
  asymptotic growth rate, in such a way that the nonzero terms come first,
  $u_1 = \cdots = u_t$
  and
  $\Re  \alpha_1 = \cdots = \Re  \alpha_t$,
  while
  \[ z^{\Re  \alpha_1} \mathrm e^{\Re  u_1(z)} \gg z^{\Re  \alpha_p}
  \mathrm e^{\Re  u_p(z)}, \quad z \to 0, z > 0 \]
  for all $p \geq t+1$ such that $y_p \neq 0$.
  Using the definition of $\pi_{i}$ and the fact that
  $\mathcal{S}_{\tuple k,\tuple d}(b_j)(z)$
  tends to $b_j(0)$ as $z\to0$ in the positive reals, it follows that 
  \[
    z^{-{\Re  \alpha_1}}\exp(- u_1(z)) \, \pi_i(\hat y)(z) 
    = \sum_{j = 1}^t c_j b_j(0) z^{\mathrm{i}\Im \alpha_j} + \mathrm{o}(1)
  \]
  (as $z \to 0$, $z >0$) for some nonzero constants~$c_j$.
  Since the $(\alpha_j,u_j)$ are pairwise distinct by assumption and the $(\Re \alpha_j,u_j)$
  are equal for $j=1,\dots,t$, the $\Im \alpha_j$ are pairwise
  distinct for $j=1,\dots,t$.

  Now assume that $\pi_i(\hat y)=0$. 
  Then, for all $\lambda > 0$, the expression
  $\sum_{j = 1}^t c_j b_j(0) (\lambda z)^{\mathrm{i}\Im \alpha_j}$
  tends to~$0$ as $z \to 0$, $z>0$.
  Choosing $\lambda=\e^p$ for $p=1,\dots,t$, it follows that if
  not all the $b_j(0)$ were zero, the $t \times t$ determinant
  \[
    \det\bigl((\mathrm e^p z)^{\mathrm{i}\Im \alpha_{q}}\bigr)
    _{p, q}
    = z^{\mathrm i \Im (\alpha_1 + \dots + \alpha_t)}
    \det\bigl((\mathrm e^{\mathrm{i}\Im \alpha_{q}})^p\bigr)
    _{p, q}
  \]
  would tend to zero as well, which however is not the case. 
  Therefore $b_j(0)=0$ for $j=1,\dots,t$, and therefore
  $\hat y_j=0$ for $j=1,\dots,t$, and therefore $\hat y_j=0$ for $j=1,\dots,\ell_i$. 

2. Let $h\in W$ be hyperexponential. Then the expansion $\hat h$ of $h$ at $z_i$
is clearly a local solution, so $\hat h\in V_{i,j}$ for some~$j$. We show that
$\pi_i(\hat h)=ch$ for some $c\in\set C$.  The map $\pi_i$ is a differential
homomorphism because $\mathcal{S}_{\tuple k,\tuple d}$ is (as remarked in
Section~\ref{sec:analytic}) and the (formal) exponential parts $\Exp(e_{i,j})$ are mapped
to analytic functions satisfying the same differential equations.
Since $h$ is hyperexponential, it satisfies a first order linear differential equation.
Since $\hat h$ is the expansion of~$h$, it satisfies the same equations as~$h$.
Since $\pi_i$ is a differential homomorphism, $\pi_i(\hat h)$ satisfies the same equations as~$\hat h$.
Hence $\pi_i(\hat h)$ and $h$ satisfy the same first-order differential equation.
The claim follows.
\end{proof}

The definition of the maps $\pi_1,\dots,\pi_m$ as outlined above relies on analytic 
continuation, a concept which is only available if $C=\set C$. For actual computations,
we must work in a computable coefficient domain. At this point, we use numerical approximations.
By van der Hoeven's result quoted in Section~\ref{sec:analytic}, we are
able to compute for every given $\hat y\in V_{i,j}$ and every given $\varepsilon>0$
a vector $Y_\varepsilon\in\set Q(\mathrm{i})^r$ with 
\[
  \bigl\Vert\bigl(D^k\pi_i(\hat y)(z_0)\bigr)_{k=0}^{r-1}-Y_\varepsilon\bigr\Vert_\infty<\varepsilon. 
\]
Using these approximations, the linear algebra parts of Algorithm~\ref{alg:1} are then
performed with ball arithmetic to keep track of accumulating errors during the
calculations.  The test in line~\ref{step:intersect} of this algorithm requires
to check whether a certain matrix has full rank.  There are two possible outcomes:
If during the Gaussian elimination we can find in every iteration an
entry which is definitely different from zero, then the rank of the matrix is
definitely maximal and the intersection of the vector spaces is definitely
empty. We are then entitled to discard the possible extension of the partial
tuple under consideration. On the other hand, if during the Gaussian elimination
we encounter a column in which all the entries are balls that contain zero, this
can either mean that the intersection is really nonempty, or that the accuracy
of the approximation was insufficient. In this case, in order to be on the safe
side, we must consider the intersection as nonempty and include the
corresponding tuple.

Regardless of which initial accuracy $\varepsilon$ is used, this variant of
Algorithm~\ref{alg:1} produces a set of tuples that is guaranteed to contain all
correct ones, but may possibly contain additional ones. With sufficiently high precision,
the number of tuples in the output that actually have an empty intersection will drop to zero.
We don't need to know in advance which precision is sufficient in this sense, because
it is not dramatic to have some extra tuples in the output as long as they are not
too many. As a pragmatic strategy balancing precision and output size, one might
start the algorithm with some fixed precision~$\varepsilon$ and let it abort and
restart with doubled precision whenever $|U|$ exceeds~$2r$, say. 

Observe that the numerical approximation is only used to determine the tuple set~$U$,
and we do not use it to somehow reconstruct the exact symbolic hyperexponential solutions
from it. We therefore don't expect to need very high precision in typical situations. 

\section{A Detailed Example}

Consider the operator
\begin{alignat*}1
  P=p_0+p_1D+p_2D^2+p_3D^3\in\set Q[x][D]
\end{alignat*}
where
\begin{center}\scriptsize
 $p_0=-105 x^{20}+3570 x^{19}-58026 x^{18}+556216
    x^{17}-3456830 x^{16}+14810744 x^{15}-45667732 x^{14}+104614932 x^{13}-182764261 x^{12}+249940430 x^{11}-276371642
    x^{10}+257839924 x^9-211785148 x^8+154714472 x^7-95675216 x^6+45214304 x^5-13863936 x^4+1685888 x^3+424960 x^2-182784
    x+20480$, \\[5pt]
 $p_1=(x-1) x (105 x^{19}-3150 x^{18}+51456 x^{17}-489796 x^{16}+2938210 x^{15}-11903624 x^{14}+34247824
    x^{13}-72603516 x^{12}+116974957 x^{11}-148046826 x^{10}+153582952 x^9-137261696 x^8+109046080 x^7-75250624
    x^6+41559168 x^5-16084864 x^4+3278080 x^3+163840 x^2-231424 x+32768)$, \\[5pt]
 $p_2=-4 (x-2)^2 (x-1)^3 x^2 (30 x^{15}-693 x^{14}+7314 x^{13}-42905 x^{12}+155930 x^{11}-378483
    x^{10}+649718 x^9-828795 x^8+820160 x^7-645092 x^6+398200 x^5-182384 x^4+54656 x^3-5696 x^2-2944 x+1024)$, \\[5pt]
 $p_3=4 (x-2)^4 (x-1)^5 x^4 (15 x^{10}-258 x^9+1492 x^8-4446 x^7+8309 x^6-10972 x^5+10520 x^4-6456
    x^3+1552 x^2+480 x-256)$. 
\end{center}
The leading coefficient $p_3$ has 13 distinct roots in~$\set C$, but those coming from the degree-10-factor turn out
to be apparent, so we can ignore them. It thus remains to study the singular points $z_1:=0$, $z_2:=1$, $z_3:=2$,
and $z_4:=\infty$. 

For each singular point, we find three linearly independent generalized series solutions with two distinct
exponential parts:
\begin{alignat*}3
  V_{1,1}&=\set C\hat y_{1,1}\qquad& V_{1,2}&=\set C\hat y_{1,2}+\set C\hat y_{1,3},\\
  V_{2,1}&=\set C\hat y_{2,1}\qquad& V_{2,2}&=\set C\hat y_{2,2}+\set C\hat y_{2,3},\\
  V_{3,1}&=\set C\hat y_{3,1}\qquad& V_{3,2}&=\set C\hat y_{3,2}+\set C\hat y_{3,3},\\
  V_{4,1}&=\set C\hat y_{4,1}\qquad& V_{4,2}&=\set C\hat y_{4,2}+\set C\hat y_{4,3}
\end{alignat*}
where
\allowdisplaybreaks
\begin{alignat*}1
  \hat y_{1,1} &
                = \exp(\tfrac1x)\Bigl(1 -\tfrac49x+\tfrac{37}{32}x^2+\tfrac{83}{384}x^3+\cdots\Bigr),\\ 
  \hat y_{1,2} &
                = \sqrt{x}\Bigl(1 - x - \tfrac{25}{24}x^3+\cdots \Bigr), \\ 
  \hat y_{1,3} &
                = \sqrt{x}\Bigl(x^2 - \tfrac74x^3+\tfrac{9}{32}x^4+\cdots\Bigr),\\[3pt] 
  \hat y_{2,1} &
                = (x-1)^3 + (x-1)^5 -\tfrac43(x-1)^6 + \cdots,\vphantom{\Bigl(}\\ 
  \hat y_{2,2} &
                = \exp(\tfrac1{x-1})\Bigl(1 + \tfrac12(x-1) + \tfrac{19}{120}(x-1)^3+\cdots\Bigr),\\
  \hat y_{2,3} &
                = \exp(\tfrac1{x-1})\Bigl((x-1)^2 + \tfrac{23}{30}(x-1)^3 +\cdots\Bigr),\\[3pt]
  \hat y_{3,1} &
                = 1 - \tfrac34(x-2) + \tfrac{39}{32}(x-2)^2 - \tfrac{673}{384}(x-2)^3 + \cdots,\vphantom{\Bigl(}\\ 
  \hat y_{3,2} &
                = \tfrac1{(x-2)^2}\exp(\tfrac1{x-2})\Bigl(1+\tfrac{11}{4}(x-2)+\cdots\Bigr),\\ 
  \hat y_{3,3} &
                = \tfrac1{(x-2)^2}\exp(\tfrac1{x-2})\Bigl((x-2)^3 + \tfrac14(x-2)^4+\cdots\Bigr),\\[3pt] 
  \hat y_{4,1} &
                = x\bigl(1 + 3x^{-1} + 9 x^{-2} + \tfrac{79}{3} x^{-3} + 74 x^{-4} + \cdots\bigr),\vphantom{\Bigl(}\\ 
  \hat y_{4,2} &
                = \sqrt x\Bigl(1 + x^{-1} + \tfrac32 x^{-2} + \tfrac{13}6x^{-3}+\cdots\Bigr),\\ 
  \hat y_{4,2} &
                = \sqrt x\Bigl(x^3 + x + \tfrac{19}6 x^{-1} + \tfrac{283}{30}x^{-2}+\cdots\Bigr). 
\end{alignat*}
Let us choose $z_0=3$ as ordinary reference point and take the branch of the logarithm for which 
$\sqrt{x}$ is positive and real on the positive real axis. 
The example was chosen in such a way that all the power series are convergent in some neighborhood
of the expansion point, so that we do not need to worry about sectors and resummation theory but can use the
somewhat simpler algorithm for effective analytic continuation in the ordinary case to compute the values
of the analytic functions $y_{i,j}:=\pi_i(\hat y_{i,j})$ ($i=1,\dots,4$; $j=1,2,3$). 
The vectors $\bigl(y_{i,j}(z_0), Dy_{i,j}(z_0), D^2y_{i,j}(z_0)\bigr)$ to five decimal digits of accuracy are as follows. 
\begin{alignat*}3
   W_{1,1}&{=}\Bigl[\begin{pmatrix} -200.15 \\ 322.46 \\ -1184.8 \end{pmatrix}\Bigr],\
 & W_{1,2}&{=}\Bigl[\begin{pmatrix} -70.513 \\ -46.308 \\ -101.17 \end{pmatrix}\!,
                  \begin{pmatrix} -156.55 \\ -91.322 \\ -205.47 \end{pmatrix}\Bigr],\\[5pt]
   W_{2,1}&{=}\Bigl[\begin{pmatrix} 30.349 \\ -48.896 \\ 179.66 \end{pmatrix}\Bigr],\
 & W_{2,2}&{=}\Bigl[\begin{pmatrix} 12.494 \\ 5.2891 \\  13.066 \end{pmatrix}\!,
                  \begin{pmatrix} 77.105 \\ 44.216 \\ 99.931 \end{pmatrix}\Bigr],\\[5pt]
   W_{3,1}&{=}\Bigl[\begin{pmatrix} .74285 \\ -.061904 \\ .14960 \end{pmatrix}\Bigr],\
 & W_{3,2}&{=}\Bigl[\begin{pmatrix} 15.580 \\ -31.307 \\ 105.26 \end{pmatrix}\!,
                  \begin{pmatrix} 4.5433 \\ 2.6503 \\ 5.9631 \end{pmatrix}\Bigr],\\[5pt]
   W_{4,1}&{=}\Bigl[\begin{pmatrix} 30.349 \\ -48.896 \\ 179.66 \end{pmatrix}\Bigr],\
 & W_{4,2}&{=}\Bigl[\begin{pmatrix} 2.8557 \\ -.23797 \\ .57510 \end{pmatrix}\!,
                  \begin{pmatrix} 63.199 \\ 41.308 \\ 90.353 \end{pmatrix}\Bigr].
\end{alignat*}
We now go through Algorithm~\ref{alg:1}. Start with the partial tuples $(1)$ and~$(2)$ corresponding
to the vector spaces $W_{1,1}$ and $W_{1,2}$, respectively. 
To compute the intersection of $W_{1,1}$ and $W_{2,1}$ we apply Gaussian elimination to the $3\times 2$-matrix
whose columns are the generators of $W_{1,1}$ and $W_{2,1}$:
\begin{alignat*}1
  \begin{pmatrix}
    -200.15 & 30.349 \\ 
     322.46 & -48.896 \\
    -1184.8 & 179.66
  \end{pmatrix} \longrightarrow
  \begin{pmatrix}
    -200.15 & 30.349 \\
            & 0.00 \\
            & 0.00     
  \end{pmatrix}
\end{alignat*}
The notation $0.00$ refers to some complex number $z$ with $|z|<5\cdot 10^{-3}$, which may or may not be zero, 
while the blank entries in the left column signify exact zeros that have been produced by the elimination.
As the remaining submatrix does not contain any entry which is certainly nonzero, we regard the intersection 
as nonempty, which in this case means $W_{1,1}=W_{2,1}$. The partial tuple $(1)$ is extended to $(1,1)$. 

The intersections $W_{1,1}\cap W_{2,2}$ and $W_{1,2}\cap W_{2,1}$ turn out to be trivial, as they have to
be if we really have $W_{1,1}=W_{1,2}$, because the sums $W_{1,1}\oplus W_{1,2}$ and $W_{2,1}\oplus W_{2,2}$
are direct. It thus remains to consider the intersection $W_{1,2}\cap W_{2,2}$. 
Applying Gaussian elimination to the $3\times 4$-matrix whose columns are the generators of $W_{1,2}$ and $W_{2,2}$,
we find
\begin{alignat*}1
  &\begin{pmatrix}
    -70.513 & -25.596 & 12.494 & 77.105 \\
    -46.308 & 2.1330 & 5.2891 & 44.216 \\ 
    -101.17 & -5.1548 & 13.066 & 99.931
  \end{pmatrix}\\
  \longrightarrow{}&
  \begin{pmatrix}
    -70.513 & -25.596 & 12.494 & 77.105 \\
            & -17.50 & 4.440 & 9.777 \\
            &       & 0.00 & 0.00 
  \end{pmatrix},
\end{alignat*}
which suggests that we have $W_{1,2}=W_{2,2}$. We extend the partial tuple $(2)$ to $(2,2)$. 
At the end of the first iteration, we have $U=\{(1,1),(2,2)\}$.

In the second iteration, we find $W_{(1,1)}\cap W_{3,1}=\{0\}$ and $W_{(1,1)}\subseteq W_{3,2}$,
so we extend the partial tuple $(1,1)$ to $(1,1,2)$ and record $W_{(1,1,2)}=W_{(1,1)}=W_{1,1}$. 
Furthermore we find $W_{3,1}\subseteq W_{(2,2)}$, so we extend $(2,2)$ to $(2,2,1)$ and record
$W_{(2,2,1)}=W_{3,1}$. Finally, there is a nontrivial intersection between $W_{(2,2)}$ and $W_{3,2}$:
\begin{alignat*}1
  &
  \begin{pmatrix}
                    12.494 & 77.105 & 15.580 & 4.5433 \\
                    5.2891 & 44.216 & -31.307 & 2.6503 \\
                    13.066 & 99.931 & 105.26 & 5.9631 
  \end{pmatrix}\\
  \longrightarrow{}&
  \begin{pmatrix}
    12.494 & 77.105 & 15.580 & 4.5433 \\
           & -27.34 & 89.53 & -1.72 \\
           &        & 216. & 0.00
  \end{pmatrix}
\end{alignat*}
suggests a common subspace of dimension 1 generated by the second listed generator of~$W_{3,2}$. 
We therefore extend the partial tuple $(2,2)$ to $(2,2,2)$ and record $W_{(2,2,2)}=[ (4.5433, 2.6503, 5.9631) ]$.
At the end of the second iteration, we have $U=\{(1,1,2),(2,2,1),(2,2,2)\}$. 

For the final iteration, we see by inspection that $W_{4,1}=W_{2,1}=W_{(1,1,2)}$, so we extend $(1,1,2)$
to~$(1,1,2,1)$. Because $\dim W_{4,1}=1$ and the sums of the vector spaces are direct, the other two partial 
tuples cannot also have a nontrivial intersection with~$W_{4,1}$, nor can $W_{(1,1,2)}\cap W_{4,2}$ be nontrivial. 
We do however have $W_{(2,2,1)}\subseteq W_{4,2}$ and
$W_{(2,2,2)}\subseteq W_{4,2}$, so the algorithm terminates with the output $U=\{(1,1,2,1),(2,2,1,2),(2,2,2,2)\}$. 

At this point we know that every hyperexponential solution of the operator $P$ must have one of the following three
exponential parts:
\begin{alignat*}1
  \frac1{(x-2)^2}\exp\Bigl(\frac1x+\frac1{x-2}\Bigr) &\qquad\text{from (1,1,2,1)}\\
  \sqrt{x}\exp\Bigl(\frac1{x-1}\Bigr)                &\qquad\text{from (2,2,1,2)}\\
  \sqrt{x}\exp\Bigl(\frac1{x-1}+\frac1{x-2}\Bigr)    &\qquad\text{from (2,2,2,2).}
\end{alignat*}
Following the steps of Algorithm~\ref{alg:main}, it remains to check whether
some rational function multiples of these terms are solutions of~$P$. The
important point is that we have to do this only for three different candidates,
while the naive algorithm would have to go through all $2^4=16$ combinations.
Indeed, it turns out that $P$ has the following three hyperexponential
solutions:
\begin{alignat*}1
  &\frac{(x-1)^3}{(x-2)^2}\exp\Bigl(\frac1x+\frac1{x-2}\Bigr),\quad
  \sqrt{x}\exp\Bigl(\frac1{x-1}\Bigr),\\
  &(x-2)x^2 \sqrt{x}\exp\Bigl(\frac1{x-1}+\frac1{x-2}\Bigr).
\end{alignat*}

\section{Concluding Remarks}

Our algorithm as described above takes advantage of the fact that series
expansions of hyperexponential terms cannot involve exponential terms with
ramification ($s>1$) or logarithms ($m>0$), by letting the morphisms~$\pi_i$ map
all these irrelevant series solutions to zero. As a result, we get smaller
vector spaces~$W_{i,j}$, which not only reduces the expected computation time
per vector space intersection but also makes it somehow more likely for
intersections to be empty, thus decreasing the chances of getting tuples that
do not correspond to hyperexponential solutions. 

As a further refinement in this direction, it would be desirable to exploit the
fact that if $\hat h=\Exp(e)b$ is the expansion of some hyperexponential
term~$h$, then the formal power series~$b$ must be convergent in some
neighborhood of the expansion point. Instead of the vector spaces $W_{i,j}$
used above, it would be sufficient to consider the subspaces
$W_{i,j}'\subseteq W_{i,j}$ corresponding to generalized series solutions
involving only convergent power series. Besides the advantage of having to work
with even smaller vector spaces, an additional advantage would be that the
numerical evaluation becomes simpler because algorithms for the regular
case~\cite{ChudnovskyChudnovsky1990,vdHoeven1999} become
applicable. Implementations of these algorithms are available~\cite{mezzarobba10}, which
to our knowledge is not yet the case for van der Hoeven's general algorithm for
the divergent case ~\cite{hoeven07}. Unfortunately however, it is not obvious how to compute from a
given basis of $W_{i,j}$ a basis of the subspace~$W_{i,j}'$.  Miller's
algorithm~\cite{wimp84} numerically solves a similar problem, but so far we have
not been able to turn the underlying convergence statements into explicit error
bounds that would yield an algorithm producing output with certified precision.

Finally, it would of course be also interesting to see an analog of our
algorithm for finding hypergeometric solutions of linear recurrence equations
with polynomial coefficients. A translation is not immediate because there is no
notion of local solution around a finite singularity in this case.

\bibliographystyle{plain}
\bibliography{bib}

\end{document}